%% file: strings.tex
\documentclass{llncs}

\usepackage{simplemargins}
\usepackage{times,amsmath,amssymb,xspace}
\usepackage{txfonts}
\usepackage{graphicx}
\usepackage{multirow}
\usepackage{amsmath}
\usepackage{amssymb}
\usepackage{theorem}
\usepackage{epsfig}
\usepackage{fancyvrb}
\usepackage{url}
\usepackage{cite}
\usepackage{multirow}
\usepackage{wrapfig}
\usepackage{comment}
\usepackage{authblk}
\usepackage{array}

\setleftmargin{1.3in}
\setrightmargin{1.3in}
\settopmargin{1.3in}
\setbottommargin{1.3in}

\numberwithin{equation}{section}	
\theoremstyle{plain}			
\newtheorem{thm}{Theorem}

\newtheorem{cor}[thm]{Corollary}

\theoremstyle{definition}		

\newtheorem{examp}[thm]{Example}

\newcommand{\vsep}{\hspace{3mm} | \hspace{3mm}}

\newcommand{\Lelr}[1]{\mathcal{L}_{e,l,r}^{#1}}
\newcommand{\Lel}[1]{\mathcal{L}_{e,l}^{#1}}

\newcommand{\Le}[1]{\mathcal{L}_{e}^{#1}}
\renewcommand{\S}{\mathcal{S}}
\newcommand{\R}{\mathcal{R}}

\begin{document}
\title{(Un)Decidability Results for Word Equations with Length and Regular Expression Constraints}

\author{Vijay Ganesh$^\dagger$, Mia Minnes$^*$, Armando Solar-Lezama$^\dagger$ and Martin
 Rinard$^\dagger$}
  \institute{$^\dagger$Massachusetts Institute of Technology\\
  \{vganesh, asolar, rinard\} @csail.mit.edu\\
  $^*$University of California, San Diego \\
  minnes@math.ucsd.edu}

\maketitle
\input{abstractNeg}
\input{introductionNeg}
\input{preliminaries}
\input{undecidabilityNeg}
\input{decidabilityQs}

\input{relatedworkNeg}



\bibliographystyle{plain}
\bibliography{strings}
\end{document}

%% file: abstractNeg.tex
\begin{abstract}
\label{sec:abstract}

\sloppypar{ We prove several decidability and undecidability results
  for the satisfiability and validity problems for 
  languages that can express solutions to word equations with length 
  constraints.
  The atomic formulas over this language are
  equality over string terms (word equations), linear inequality over the
  length function (length constraints), and membership in
  regular sets. These questions are important in
  logic, program analysis, and formal verification.
  Variants of these questions have been studied for
many decades by mathematicians. More recently, practical
satisfiability procedures (aka SMT solvers) for these formulas have
become increasingly important in the context of security analysis for
string-manipulating programs such as web applications.

  We prove three main theorems. First, we give a new proof of
  undecidability for the validity problem for the set of sentences
  written as a $\forall\exists$ quantifier alternation applied to
  positive word equations. A corollary of this undecidability result
  is that this set is undecidable even with sentences with at most two
  occurrences of a string variable. Second, we consider Boolean
  combinations of quantifier-free formulas constructed out of word
  equations and length constraints. We show that if word equations can
  be converted to a {\it solved form}, a form relevant in practice,
  then the satisfiability problem for Boolean combinations of word
  equations and length constraints is decidable. Third, we show that
  the satisfiability problem for quantifier-free formulas over word
  equations in {\it regular solved form}, length constraints, and the
  membership predicate over regular expressions is also decidable.}

\end{abstract}

%% file: introductionNeg.tex
\section{\bf Introduction} 

The complexity of the satisfiability problem for formulas over finite-length strings
(theories of strings) has long been studied, including by 
Quine~\cite{Quine}, Post, Markov
and Matiyasevich~\cite{Matiyasevich}, Makanin~\cite{makanin}, and
Plandowski~\cite{KarhumakiPM97, plandowski99, plandowski2006}. While
much progress has been made, many questions 
remain open especially when the language is enriched with new predicates.

Formulas over strings have become important in the
context of automated bugfinding~\cite{hampi, prateek}, and analysis of
database/web applications~\cite{emmiMS2007, rupak,
  WassermannSu2007}. These program analysis and bugfinding tools
read string-manipulation programs and generate formulas expressing
their results.  These formulas
contain equations over string constants and variables, membership
queries over regular expressions, and inequalities between string lengths.
In practice, formulas of this form have been solved by
off-the-shelf satisfiability procedures such as HAMPI~\cite{hampi2,
  hampi} or Kaluza~\cite{prateek}. In this context, a deeper
understanding of the theoretical aspects of the satisfiability problem
for this class of formulas may be useful in practice. 

\vspace{0.1cm}
\noindent{\bf Problem Statement:} We address three problems. First, what is
a boundary for decidability for fragments of the theory of word equations?
Namely, is the
$\forall\exists$-fragment of the theory of word equations
decidable? Second, is the satisfiability problem for quantifier-free
formulas over word equations and the length function decidable under
some minimal practical conditions? Third, is the satisfiability
problem for quantifier-free formulas over word equations, the length
function, and regular expressions decidable under some minimal
practical conditions?

The question of whether the
satisfiability problem for the quantifier-free theory of word equations
and length constraints is decidable has remained open for several
decades.  Our decidability results are a partial and
conditional solution. Matiyasevich~\cite{matiyasevich2008} observed the
relevance of this question to a novel resolution of Hilbert's Tenth
Problem.  In particular, he showed that if the satisfiability problem
for the quantifier-free theory of word equations and length
constraints is undecidable, then it gives us a new way to prove
Matiyasevich's Theorem (which resolved the famous problem)~\cite{Matiyasevich, matiyasevich2008}.

\vspace{-0.2cm}
\subsection*{Summary of Contributions:}  

\begin{enumerate}

\item We show that the validity problem (decision problem) for the
  set of sentences written as a $\forall\exists$ quantifier
  alternation applied to positive word equations (i.e., AND-OR
  combination of word equations without any negation) is
  undecidable. (Section~\ref{sec:undecidability})

\item We show that if word equations can be converted to a {\it solved
  form} then the satisfiability problem for Boolean combinations of
  word equations and length constraints is
  decidable. (Section~\ref{sec:decidability})

\item The above-mentioned decidability result has immediate practical
  impact for applications such as bug-finding in JavaScript and PHP
  programs. We empirically studied the word equations in the formulas
  generated by the Kudzu JavaScript bugfinding tool~\cite{prateek} and
  verified that most word equations in such formulas are either
  already in solved form or can be automatically and easily converted
  into one. (Section~\ref{sec:decidability}). We further show that the
  satisfiability problem for quantifier-free formulas constructed out
  of Boolean combinations of word equations in {\it regular solved
    form} with length constraints and the membership predicate for
  regular sets is also decidable. This is the first such decidability
  result for this set of
  formulas. (Section~\ref{sec:elr-decidability})
\end{enumerate}

We now outline the layout of the rest of the paper.
In Section~\ref{sec:prelim} we define
a theory of word equations, length constraints, and regular
expressions. In Section~\ref{sec:undecidability} we prove the
undecidability of the theory of $\forall\exists$ sentences over
positive word equations. In Section~\ref{sec:decidability}
(resp. Section~\ref{sec:elr-decidability}) we give a conditional
decidability result for the satisfiability problem for the
quantifier-free theory of word equations and length constraints
(resp. word equations, length constraints, and regular
expressions). Finally, in Section~\ref{sec:relwork} we provide a
comprehensive overview of the decidability/undecidability results for
theories of strings over the last several decades.

%% file: preliminaries.tex
\section{Preliminaries}
\label{sec:prelim}

\subsection{Syntax}

\noindent{\bf Variables:} We fix a disjoint two-sorted set of
variables $var = var_{str} \cup var_{int}$; $var_{str}$
consists of string variables, denoted $X,Y,S, \ldots$ and $var_{int}$
consists of integer variables, denoted $m,n,\ldots$.

\noindent{\bf Constants:} We also fix a two-sorted set of constants
$Con = Con_{str} \cup Con_{int}$.  Moreover, $Con_{str} \subset
\Sigma^*$ for some finite alphabet, $\Sigma$, whose elements are
denoted $f, g, \ldots$. Elements of $Con_{str}$ will be referred to as
{\it string constants} or {\it strings}. Elements of $Con_{int}$ are
nonnegative integers. The empty string is represented by $\epsilon$.

\noindent{\bf Terms:} Terms may be string terms or
length terms. A string term ($t_{str}$ in Figure~\ref{fig:syntax}) is
either an element of $var_{str}$, an element of $Con_{str}$, or a
concatenation of string terms (denoted by the function $concat$ or
interchangeably by $\cdot$).  A length term ($t_{len}$ in
Figure~\ref{fig:syntax}) is an element of $var_{int}$, an
element of $Con_{int}$, the length function applied to a string
term, a constant integer multiple of a length term, or a sum of
length terms.

\noindent{\bf Atomic Formulas:} There are three types of atomic formulas: 
(1) word equations ($A_{wordeqn}$), 
(2) length constraints ($A_{length}$), or 
(3) membership in a set defined by a 
regular expression ($A_{regexp}$).
Regular expressions are defined inductively, 
where constants and the empty string form the base case, and the operations
of concatenation, alternation, and Kleene star are used to build up 
more complicated expressions
(see details in~\cite{ullmanbook}).
Regular expressions may not contain variables.

\noindent{\bf Formulas:} Formulas are defined inductively 
over atomic formulas (see Figure~\ref{fig:syntax}). We include
quantifiers of two kinds: over  string variables and over integer variables.

\noindent{\bf Formula Nomenclature:} We now establish notation for the
classes of formulas we will analyze.  
Define $\Lelr{1}$ to be the
first-order two-sorted language over which the formulas described
above (Figure~\ref{fig:syntax}) are constructed. This language
contains word equations, length constraints, and membership in 
given regular sets. The superscript $1$ in $\Lelr{1}$ denotes
that this language allows quantifiers, and the subscripts $l,e,r$
stand for ``length", ``equation", and ``regular expressions"
(respectively). Let $\Lel{1}$ be the analogous set of
first-order formulas restricted to word equations and length constraints
as the only atomic formulas,
and let $\Le{1}$ be the collection of formulas whose
only atomic formulas are word equations. 
Define $\Lelr{0}$ to be the set of quantifier-free
$\Lelr{1}$ formulas.  Similarly, $\Lel{0}$ and $\Le{0}$ are the
quantifier-free versions of $\Lel{1}$ and $\Le{1}$,
respectively. 

Recall that a formula is in {\it prenex normal form} if all quantifiers
appear at the front of the expression: that is, the formula has a string of 
quantifiers and then a Boolean combination of atomic formulas.  
It is a standard result (see, for example~\cite{EbFlumThomas})
that any first-order formula can be translated into
prenex normal form.  We therefore assume that all formulas are given in 
this form.  Intuitively, a variable is {\it free} in a
formula if it is not quantified. For example, in the formula $\forall y \phi(y,x)$,
the variable $y$ is {\it bound} while $x$ is {\it
  free}. For a full inductive definition, see~\cite{EbFlumThomas}. A
formula with no free variables is called a {\it sentence}.

%
\begin{figure}[t!]
\[
\begin{array}{lll}
  F & \Coloneqq & Atomic  \hspace{1.2mm} \vsep F \wedge F  \vsep F \vee F  \vsep \neg F\\
  & \vsep & \exists x.F(x)  \vsep \forall x.F(x)\\
  Atomic & \Coloneqq & A_{wordeqn}  \hspace{0.1mm} \vsep A_{length}  \vsep A_{regexp}\\
  A_{wordeqn} & \Coloneqq & t_{str} = t_{str}\\
  A_{length} & \Coloneqq & t_{len} \leq c \hspace{40mm} \text{where } c \in Con_{int}\\
  A_{regexp} & \Coloneqq & t_{str} \in RE \hspace{37mm} \text{ where RE is a regular expression}\\
  t_{str} & \Coloneqq & a \vsep X \vsep concat(t_{str},...,t_{str}) \hspace{9.5mm} \text{where} \hspace{1mm} a \in Con_{str} \hspace{1mm} \& \hspace{1mm} X \in var_{str}\\
  t_{len} & \Coloneqq & m \vsep v \vsep len(t_{str}) \vsep \Sigma_{i=1}^{n}c_i * t^i_{len} \hspace{3mm} \text{where} \hspace{1mm} m, n, c_i \in Con_{int} \hspace{1mm} \& \hspace{1mm} v \in var_{int}\\
\end{array}
\]
\caption{\label{fig:syntax} The syntax of $\Lelr{1}$-formulas.}
\end{figure}

\subsection{Semantics and Definitions}

For a word, $w$, $len(w)$ denotes the length of $w$. For a word
equation of the form $t_1 = t_2$, we refer to $t_1$ as the left hand
side (LHS), and $t_2$ as the right hand side (RHS). 

We fix a string alphabet, $\Sigma$. 
Given an $\Lelr{1}$
formula $\theta$, an {\it assignment} for $\theta$ (with respect to $\Sigma)$ is
a map from the set of free variables in $\theta$ to 
$\Sigma^* \cup \mathbb{N}$ (where string variables are 
mapped to strings and integer variables are mapped to numbers).
Given such an assignment, $\theta$ can be interpreted as an assertion
about $\Sigma^*$ and $\mathbb{N}$.   If this assertion is true, then 
we say that $\theta$ itself is {\it true} under the assignment.  
If there is some assignment which makes $\theta$ true, then
$\theta$ is called {\it satisfiable}.
An $\Lelr{1}$-formula with no satisfying assignment
is called an {\it unsatisfiable} formula. We say two formulas $\theta,
\phi$ are {\it equisatisfiable} if $\theta$ is satisfiable iff $\phi$
is satisfiable.  Note that this is a broad definition: equisatisfiable
formulas may have different numbers of assignments and, in fact, 
need not even be from the same language.

The {\it satisfiability problem} for a set $S$ of formulas is the
problem of deciding whether any given formula in $S$ is satisfiable or
not. We say that the satisfiability problem for a set $S$ of formulas
is decidable if there exists an algorithm (or {\it satisfiability
  procedure}) that solves its satisfiability problem. Satisfiability
procedures must have three properties: soundness, completeness, and
termination. Soundness and completeness guarantee that the procedure
returns ``satisfiable" if and only if the input formula is indeed
satisfiable. Termination means that the procedure halts on all
inputs. In a practical implementation, some of these requirements may be
relaxed for the sake of improved typical performance.

Analogous to the definition of the satisfiability problem for
formulas, we can define the notion of the {\it validity problem} (aka
decision problem) for a set $Q$ of sentences in a language $L$. The
validity problem for a set $Q$ of sentences is the problem of
determining whether a given sentence in $Q$ is true under all assignments. 


\subsection{\bf Representation of Solutions to String Formulas}

It will be useful to have compact representations of sets of solutions
to string formulas.  For this, we use Plandowski's terminology of {\it
  unfixed parts} \cite{plandowski2006}.  Namely, fix a set of new
variables $V$ disjoint from all of $\Sigma$, $Con$, and $var$.  For $\theta$
an $\Lelr{1}$ formula, an {\it assignment with unfixed parts} is a
mapping from the free variables of $\theta$ to string elements of the
domain or $V$.  Such an assignment represents the family of solutions
to $\theta$ where each element of $V$ is consistently replaced by a
string element in the domain.  (See example \ref{ex:unfixed} below.)

Another tool for compactly encoding many solutions to a formula
 is the use of {\it integer parameters}.  If $i$
is a non-negative integer, we write $u^{i}$ to denote the $i$-fold
concatenation of the string $u$ with itself.  An {\it assignment with
  integer parameters} to the formula $\theta$ is a map from the free
variables of $\theta$ to string elements of the domain, perhaps with
integer parameters occurring in the exponents.  (See example \ref{ex:abba} below.)

Combining these two representations, we also consider assignments with unfixed parts and
integer parameters.  These assignments will provide the general
framework for representing solution sets to $\Lelr{1}$ formulas
compactly.

\subsection{\bf Examples}

We consider some sample formulas and their solution sets.  
The string alphabet is $\Sigma = \{a,b\}$.  (Many of
the examples in this paper are from existing literature by Plandowski et
al.~\cite{plandowski2006}.)

\begin{examp}\label{ex:unfixed}
  Consider the $\Le{0}$ formula which is a word equation $X=aYbZa$ with three
  variables ($X,Y,Z$) and two string constants ($a,b$).
  The set of all solutions to this equation is described by the
  assignment $X \mapsto aybza, Y \mapsto y, Z \mapsto z$, where $V =
  \{y,z\}$ is the set of unfixed parts.
  Any choice of $y,z \in \Sigma^*$ yields a solution to the equation.
\end{examp}

\begin{examp}\label{ex:abba}
  Consider the equation $abX = Xba$ with one variable $X$. This is a
  formula in $\Le{0}$. The map $X \mapsto aba$ is a solution.
  The map $X \mapsto (ab)^{i} a$ with $i \geq 0$ is also an assignment 
  which gives a solution.  In fact, this assignment (with integer parameters)
  exactly describes all possible solutions of the word equation. 
\end{examp}

\begin{examp}
  Consider the $\Lelr{0}$ formula 
  $$abX = Xba \wedge X \in
  (ab \mid ba)(ab)^*a \wedge len(X) \leq 5.$$ The two solutions to this
  formula are $X = aba$ and $X = ababa$.
\end{examp}

%% file: undecidabilityNeg.tex
\section{\bf The Undecidability Theorem}
\label{sec:undecidability}

In this section we prove that the validity problem for the set of
$\Le{1}$ sentences over positive word equations (AND-OR combinations of
word equations) whose prenex normal form has 
$\forall\exists$ as its quantifier prefix is undecidable.

\subsection{\bf Proof Idea}

We do a reduction from the halting problem for two-counter machines,
which is known to be undecidable~\cite{ullmanbook}, to the problem in
question. To do so, we encode computation histories as
strings. The choice of two-counter machine
makes this proof cleaner than other undecidability proofs for this set
of formulas (see Section \ref{sec:relwork} for a comparison with earlier work).
The basic proof strategy is as
follows: given a two-counter machine $M$ and a finite string $w$, we
construct an $\Le{1}$ sentence $\forall S \exists S_1,\ldots, S_4
\theta(S,S_1,\ldots, S_4)$ such that $M$ does not halt on $w$ iff
this $\Le{1}$ sentence is valid.  By the construction of $\theta$, 
this will happen exactly when
all assignments to the string variable $S$ are not codes for halting
computation histories of $M$ over $w$. 
The variables $S_1,\ldots,S_4$
are used to refer to substrings of $S$ and the quantifier-free 
formula~$\theta$ expresses the property of $S$ not coding a 
halting computation history.

\subsection{\bf Recalling Two-counter Machines}

A {\it two-counter machine} is a deterministic machine which has a
finite-state control, two semi-infinite storage tapes, 
and a separate
read-only semi-infinite input tape. All tapes have a
left endpoint and no right endpoint. All tapes
are composed of cells, each of which may store a symbol from the
appropriate alphabet (the alphabet of the storage tapes is $\{Z, \text{blank}\}$;
the alphabet of the input alphabet is some fixed finite set).
 The input to the machine is a finite
string written on the input tape, starting at the leftmost cell. A
special character follows the input string on the tape to mark the end
of the input.
Each tape has a corresponding tape-head that may
move left, move right, or stay put. 
The input tape-head cannot move past the
right end of the input string.
The initial position of all the tape-heads is the leftmost
cell of their respective tapes. At each point in the computation, the 
cell being scanned by each tape-head is called that tape's {\it current cell}.

The symbol $Z$ serves as a {\it bottom of stack} marker on the storage
tapes.  Hence, it appears initially on the cell scanned by the tape
head and may never appear on any other cells. 
A non-negative integer $i$
can be represented on the storage tape by moving the tape head $i$
cells to the right of $Z$. A number stored on the storage tape can be
incremented or decremented by moving the tape-head to the right or to the left. We
can test whether the number stored in one of the storage tapes is zero by checking 
if the contents of the current cell of that tape is $Z$.
But, the equality of two numbers stored on the storage tapes cannot be directly
tested. It is well known that the two-counter machine can simulate an
arbitrary Turing machine. Consequently, the halting problem for
two-counter machines is undecidable~\cite{ullmanbook}.

More formally, a two-counter machine $M$ is a tuple $\langle Q,
\Delta, \{Z, b, c\}, 
\delta, q_0, F \rangle$ where,

\begin{itemize}
\item $Q$ is the finite set of control states of $M$, $q_0 \in Q$ is
  the initial control state, and $F \subseteq Q$ is the set of final
  control states.

\item $\Delta$ is the finite alphabet of the input tape, 
  $\{Z,b\}$ and $\{Z,c\}$ are the storage tape alphabets for the first and second tapes,
  respectively. (The distinct blank symbols for the two tapes are a notational 
  convenience.)

\item $\delta$ is the transition function for the control of
  $M$.  This function maps the domain,
    \mbox{$Q \times \Delta \times \{Z, b\} \times \{Z,c\}$} 
into 
   \mbox{$ Q \times \{in,stor1,stor2\}
    \times \{L,R\}$}.
In words, given a control state and the contents of the current
cell of each tape, the transition function specifies the next state of
the machine, a tape-head (input or one of the storage tapes) to move,
and whether this tape-head moves left ($L$) or right ($R$).
\end{itemize}

\subsection{\bf Instantaneous Description of Two-counter Machines as
  Strings}
\label{sec:ID}
We define {\it instantaneous descriptions} (ID) of two-counter
machines in terms of strings. Informally, the ID of a machine
represents its {\it entire configuration} at any instant in terms of
machine parameters such as the current control state, current
input-tape letter being read by the machine, and current storage-tape
contents.  The set of IDs will be determined both by the machine 
and the given input to the machine.

\noindent{\bf Definition of ID:} An instantaneous description (ID) of
a computation step of
a two-counter machine $M$ running on input $w$ is the
 concatenation of the following components.

\begin{itemize}

\item Current control state of $M$: represented by a character over
  the finite alphabet $Q$.

\item The input $w$ and an encoding of the current input tape
  cell.  The encoding uses string constants to represent the 
  integers between $0$ and $|w|-1$; let $N_i$ denote the string
  constant encoding the number $i$.
\
\item The finite distances of the two storage heads from the symbol
  $Z$, represented as a string of blanks (i.e., in unary
  representation). For convenience, we will use the symbol $b$ to
  denote the blanks on storage tape 1, and $c$ on storage tape 2.
\end{itemize}

Each component of an ID is separated from the others by an appropriate
special character. In what follows, we will suppress
discussion of this separator and we will assume that it is
appropriately located inside each ID. 
A lengthy but technically trivial modification of our reduction 
formula could be used to allow for the case where this separator is missing.

\noindent{\bf Definition of Initial ID:} 
For any two-counter machine $M$ and each input $w$, there is exactly one initial
ID, denoted $Init_{M,w}$.
This ID is the result of concatenating the string representations of the following data:
Initial state $q_0$ of $M$, $w$, $0$, $\epsilon$, $\epsilon$.  The ``$0$" says that the current cell
of the input tape contains the $0$th letter of $w$.  The two
``$\epsilon$"s represent the contents of the 
two storage tape: both are empty at this point.

\noindent{\bf Definition of Final ID:} We use the standard
convention that a two-counter machine halts only after
the storage tapes contain the 
unary representation of the number $0$ and the input tape-head has
moved to the leftmost position of its tape.
The ID of the machine at the end of a computation is therefore the 
concatenation of representations of $q_f, w, 0, \epsilon, \epsilon$, 
where $q_f$ is one of the
finitely many final control states $q_f \in F$ of $M$.  Observe that there are only finitely many
Final IDs.

\subsection{\bf Computation History of a Two-counter Machine as a
  String}

A {\it well-formed computation history} of a two-counter machine $M$
as it processes a given input $w$ is the concatenation of a sequence of IDs
separated by the special symbol $\#$. 
The first ID in the sequence is the initial ID of $M$ on $w$, and 
for each $i$,
$ID_{i+1}$ is the result of
transforming $ID_i$ according to the transition function of $M$. A
well-formed computation history of the machine $M$ on the string $w$
is called {\it accepting} if it is a finite string whose last ID is a
Final ID of $M$ on $w$. The last ID of a string is defined to be the
rightmost substring following a separator $\#$. 
If a finite computation history is not accepting, it is either 
not well-formed or rejecting.

\subsection{\bf Alphabet for String Formulas and The Universe of
  Strings}

Given a two-counter machine $M$ and an input string $w$, we define the
associated finite alphabet 
$$\Sigma_0 = \{\#q_i N_j w
: q_i \in Q, 0  \leq j <|w|\}.$$
This alphabet includes all possible {\it initial segments} of IDs,
not including the data about the contents of the storage tapes.
We also define  $\Sigma_1 = b$ and $\Sigma_2 = c$. We define the alphabet of strings
as $\Sigma \equiv \{\Sigma_0 \cup \Sigma_1 \cup \Sigma_2\}$, and the
universe of strings as $\Sigma^*$.
Thus, each valid ID will be in the regular set $\Sigma_0 \Sigma_1^* \Sigma_2^*$.

\subsection{\bf The Undecidability Theorem}

\begin{thm}
  The validity problem for the set of $\Le{1}$ sentences over positive
  word equations with $\forall\exists$ quantifier alternation is
  undecidable.
\end{thm}

\begin{proof} 
  \noindent{\bf By Reduction:} We reduce the halting problem for
  two-counter machines to the decision problem in question. Given a
  pair $\langle M, w \rangle$ of a two-counter machine $M$ and an
  arbitrary input $w$ to $M$, we construct an $\Le{1}$-formula
  $\theta_{M,w}(S,S_1,\ldots,S_4, U, V)$ 
  which describes the conditions for $S_1, \ldots, S_4$ to be 
  substrings of $S$ and $S$ to fail to code an accepting computation 
  history of $M$ over $w$.
  Thus,  
  \[
  \forall S \exists S_1,S_2, S_3, S_4, U, V~\left( \theta_{M,w}(S,S_1,\cdots,S_4, U, V) \right)
  \]
  is valid if and only if it is not the case that $M$ halts and accepts on $w$.
For brevity, we write $\theta$
for $\theta_{M,w}$. 

\subsubsection{\bf Structure of $\theta$:} We will define $\theta$ as the disjunction of 
ways in which $S$ could fail to encode an accepting computation history:
either $S$ does not start with the Initial
ID, or $S$ does not end with any of the Final IDs, or $S$ is not a well-formed
sequence of IDs,  or it does not follow the transition function of $M$ over
$w$.

\begin{align*}
\theta =  &(\bigvee_{E \in \text{NotInit}} S = E \cdot S_1) \vee (\bigvee_{E \in \text{NotFinal}} S = S_1 \cdot E) \vee \\
  &\text{NotWellFormedSequence}(S,S_1,\cdots,S_4) \vee \\
  &((S = S_1 \cdot S_2 \cdot S_3 \cdot S_4) \wedge (Ub = bU) \wedge (Vc = cV) \wedge \neg
  \text{Next}(S,S_1,S_2, S_3,S_4,U,V))
\end{align*}
Note that the variables $S_i$ ($i=1, \ldots, 4$) represent substrings of $S$.

\begin{itemize}
\item  {\bf NotInit and NotFinal:} The set $\text{NotInit}$ is a finite set of string constants for strings with length at most that
of the Initial ID $Init_{M,w}$ which are not equal to
$Init_{M,w}$. 
Similarly, $\text{NotFinal}$ is a set of string constants for strings that that are not equal to any of the Final
IDs, but have the same or smaller length. 

\item {\bf NotWellFormedSequence}: This subformula asserts that $S$
  is not a sequence of IDs. 
  Recall that, by definition, the set of well-formed IDs is described by the regular expression
  $\Sigma_0\Sigma_1^*\Sigma_2^* = \Sigma_0 b^* c^*$, where strings in $\Sigma_0$ 
  (as defined above) include the ID separator $\#$
  as well as codes for the control state, $w$, and letter of $w$ being scanned.
  A well-formed sequence of IDs is a string
  of the form $(\Sigma_0b^*c^*)^* - \epsilon$.
  Thus, the set described by {\bf NotWellFormedSequence} should be
  $\Sigma^* - (\Sigma_0 b^* c^*)^*$.
  In fact, we can characterize this regular set entirely in terms 
  of word equations: a string over $\Sigma = \Sigma_0 \cup \{b,c\}$ is not a well-formed sequence of IDs
  if and only if it starts with $b$
  or  $c$, or contains $cb$. 
  The fact that a non well-formed sequence
  may start with $b$ or $c$ is already captured by the
  $\text{NotInit}$ formula above. The fact that a non well-formed
  sequence contains $cb$ or is an $\epsilon$ is guaranteed by the
  following formula NotWellFormedSequence():
    $$(S = \epsilon) \vee (S = S_1 \cdot c \cdot b \cdot S_4).$$

\item{\bf Next}:

  $Next()$ asserts that the pair of variables $S_2,S_3$ form a legal
  transition.
  It is a disjunction over all (finitely many)
  possible  pairs of IDs defined by the transition function:
  \begin{align*}
    &\bigvee_{(q_2, d, g_1, g_2, q_3, t, m)\in \delta; 0 \leq n_2, n_3< |w|} S_2 = \#q_2N_{n_2}wUV \wedge S_3 = \#q_3N_{n_3}wf(U)g(V)
  \end{align*}
where $d = w(n_2)$;
$g_1 = Z$ if $U = \epsilon$ and $g_1=b$ otherwise; 
$g_2 = Z$ if $V = \epsilon$ and $g_2 = c$ otherwise;
and $f(U), g(V), N_{n_3}$ are the results of modifying the stack contents represented by $U,V$ and input tape-head
position
according to whether the value of $t$ is $in, stor1, $ or $stor2$ and whether $m$ is $L$ or $R$.
 Note that the disjunction is finite and is determined by the transition function and $w$.
 Also note that each of $\# q_2 N_{n_2} w$ and $\#q_3 N_{n_3} w$ is a single letter in $\Sigma_0$.
  \end{itemize}

\noindent{\bf Simplifying the formula:} 
The formula $\theta$ contains negated equalities in the subformula 
$\neg Next$.  However, 
each of these may be replaced by a disjunction of equalities
because $Q, |w|, \delta$ are each finite.
Hence, we can translate $\theta$ to a formula containing only conjunctions and
disjunctions of
positive word equations. 
We also observe that the formula we constructed in
the proof can be easily converted to a formula which has at most two
occurrences of any variable~\footnote{We thank Professor
  Rupak Majumdar for observing this and other improvements.}. Thus, we
get the final theorem.
\end{proof}

\begin{thm}
\label{thm:final}
The validity problem for the set of $\Le{1}$ sentences with
$\forall\exists$ quantifier alternation over positive word equations,
and with at most two occurrences of any variable, is undecidable.
\end{thm}

\noindent{\bf Bounding the Inner Existential Quantifiers:} Observe
that in $\theta$ all the inner quantifiers $S_1,\cdots,S_4, U, V$ are
bounded since they are substrings of $S$. The length 
function, $len(S_i) \leq len(S)$, can be used to bound these quantifiers.

\begin{cor}
The set of $\Lel{1}$ sentences with a
single universal quantifier followed by a block of inner bounded
existential quantifiers is undecidable.
\end{cor}


%% file: decidabilityQs.tex
\section{\bf Decidability Theorem}
\label{sec:decidability}

In this section we demonstrate the existence of an algorithm deciding
whether any $\Lel{0}$ formula has a satisfying assignment, under a
minimal and practical condition.

\subsection{\bf Word Equations and Length Constraints}

Word equations by themselves are
decidable~\cite{plandowski2006}.
Also, systems of inequalities
over integer variables are decidable
because these are expressible
as quantifier-free formulas in the language of Presburger arithmetic
and Presburger arithmetic is known
to be decidable~\cite{PRES27}. In this section, we show that if word
equations can be converted into {\em solved form}, the satisfiability
problem for quantifier-free formulas over word equations and length
constraints (i.e., $\Lel{0}$ formulas) is decidable. Furthermore, we
describe our observations of word equations in formulas generated by
the Kudzu JavaScript bugfinding tool~\cite{prateek}.  
In particular, we saw that these equations
either already appeared in solved form or could be
algorithmically converted into one.

\subsection{What is Hard about Deciding Word Equations and Length
  Constraints?}

The crux of the difficulty in establishing an unconditional
decidability result is that
it is not known whether the length constraints implied by a set of word
equations have a finite representation~\cite{plandowski2006}. 
In the case when the implied constraints do have a finite representation, 
we look for a satisfying assignment to both the implied and explicit constraints.
Such a solution can be translated into a satisfying assignment of the
word equations when the implied constraints of the system of equations
is equisatisfiable with the system itself.

\subsection{\bf Definition of Solved Form}

A word equation $w$ has a {\it solved form} if there is a finite set
$\S$ of formulas (possibly with integer parameters) that is logically
equivalent to $w$ and satisfies the following
conditions.\footnote{The idea of solved form is well known in
  equational reasoning, theorem proving, and satisfiability procedures
  for rich logics (aka SMT solvers).}

\begin{itemize}
\item Every formula in $\S$ is of the form $X = t$, where $X$ is a
  variable occurring in $w$ and $t$ is the result of finitely many
  concatenations of constants in $w$ (with possible integer
  parameters) and possible unfixed parts.  
  (Recall the definitions for integer parameters and unfixed parts
  from Section \ref{sec:prelim}.)  All integer parameters $i$ in $\S$
  are linear, of the form $ci$ where $c$ is an integer
  constant.

\item Every variable in $w$ occurs exactly once on the LHS of an
  equation in $\S$ and never on the RHS of an equation in $\S$.
\end{itemize}

The solved form corresponding to $w$ is the conjunction of
all the formulas in $\S$, denoted $\wedge \S$. If there is an algorithm
which converts any given word equation to solved form (if one
exists, and halts in finite time otherwise), and if 
$\wedge \S$ is the output of this algorithm when given
$w$, we say that the {\it effective solved form} of $w$ is $\wedge
\S$.
Solved form equations can have integer parameters,
whereas $\Lel{0}$ formulas cannot.
The solved
form is used to extract all necessary and sufficient length
information {\it implied by $w$}. 

\begin{examp} {\bf Satisfiable Solved Form Example:} Consider the
  system of word equations
\[
Xa = aY \wedge Ya = Xa.
\]
This formula can be converted into solved form as follows:
\[
X = a^i \wedge Y = a^i \qquad (i\geq0).
\]
\end{examp}

\begin{examp} {\bf Unsatisfiable Solved Form Example:} Consider the
  formula $$abX = Xba \wedge X = abY \wedge len(X) < 2$$ with
  variables $X,Y$. The set of solutions to the equation $abX = Xba$ is
  described by the map $X \mapsto (ab)^{i} a$ with $i \geq 0$ (recall Example \ref{ex:abba}). Hence
  the solved form for the system of two equations is:
  \[
  X = (ab)^{i}a \wedge Y = (ab)^{i-1} a \qquad (i > 0)
  \]
  The length constraints implied by this system are
  \[
  len(X) = 2c+1 \wedge len(Y) = 2c-1 \wedge len(X) <2 \qquad \qquad (c >0).
  \]
  This is unsatisfiable.  Hence, the original formula is also
 unsatisfiable.
\end{examp}

\begin{examp} {\bf Word Equations Without a Solved Form:}
  Not all word equations can be written in solved form. Consider the
  equation $$XabY = YbaX.$$  The map $X \mapsto a, Y \mapsto aa$ 
  is a solution, as is $X \mapsto bb, Y \mapsto b$.  However, it is known that the solutions to this
  equation cannot be expressed using linear integer
  parameters~\cite{plandowski2006}.  Thus, not all satisfiable systems of 
  equations can be expressed in solved form.
\end{examp}

\subsection{\bf Why Solved Form?} 

For word equations with an equivalent solved form, 
all length
information implied by the word equations can be represented in a
finite and {\it complete} (defined below) manner. The completeness
property enables a satisfiability procedure to decouple the word
equations from the (implied and given) length constraints, because it
guarantees that the word equation is equisatisfiable with the implied
length constraints. Furthermore, solved form guarantees that the
implied length constraints are linear inequalities, and hence their
satisfiability problem is decidable~\cite{PRES27}. This insight forms
the basis of our decidability results. It is noteworthy that most word
equations that we have encountered in practice~\cite{prateek} are
either in solved form or can be automatically converted into one.

\subsection{\bf Proof Idea for Decidability}

Without loss of generality, we consider formulas that are the conjunction of word
equations and length constraints.  (The result can be easily extended to
arbitrary Boolean combination of such formulas.) Let $\phi \wedge
\theta$ be an $\Lel{0}$-formula, where $\phi$ is a conjunction of word
equations and $\theta$ is a conjunction of length constraints.
Observe that $\phi$ implies a certain set of length constraints.

\begin{examp}
Consider the equation $X = abY$.
We have the following set $\R$ of implied length constraints: 
$$\{len(X) = 2 + len(Y), len(Y) \geq 0\}.$$
The set $\R$ is finite but exhaustive. 
That is, any other length constraint implied by the equation
$X=abY$ is either in $\R$ or is implied by $\R$ .
Consider the $\Lel{0}$ formula
$$X = abY \wedge len(Y) >1,$$
Note that $X=abY$ is satisfiable, say by the assignment with 
unfixed parts $X \mapsto aby, Y \mapsto y$.  It remains to check whether there is a solution
(represented by some choice of the unfixed part) which satisfies the length constraints
$\R \cup \{len(Y) > 1\}$.  A solution to the set of integer inequalities is $len(X) = 4, len(Y) = 2$.
Translating this to a solution of the original formulas amount to ``back-solving" for 
the exponent of unfixed parts in the solution to the word equation.
That is, since $X \mapsto aby, Y\mapsto y$ is a satisfying assignment, we can 
pick any  string of length $2$ for $y$: say, $X \mapsto abab, Y \mapsto ab$. 

Taking this example further, consider the $\Lel{0}$ formula
$$X = abY \wedge len(Y) >1 \wedge len(X) \leq 2.$$
The set of length constraints is now: 
$\{len(X) = 2 + len(Y), len(Y) \geq 0, len(Y) > 1, len(X) \leq 2\}$.
This is not satisfiable, so neither is the original formula.
\end{examp}

The set of implied length constraints for word equations that
have a solved form is also finite and exhaustive. We prove this fact
below, and use it to prove that a sound, complete and
terminating satisfiability procedure exists for $\Lel{0}$ formulas
with word equations in solved form.

\noindent{\bf Definitions:} We say that a set $\R$ of length
constraints is {\em implied by a word equation} $\phi$ if the lengths
of the strings in any solution of $\phi$ satisfy all constraints in
$\R$.  
And, $\R$  is {\it complete} for $\phi$ if any length constraint implied
by $\phi$ is either in $\R$ or is implied by a subset of $\R$. These
definitions can be suitably extended to a Boolean combination of word
equations.

%
\subsection{\bf Decidability Theorem}

We prove a set of lemmas culminating in the decidability theorem.

\begin{lemma}\label{lm:genR}
  If a word equation $w$ has a solved form $\S$, then there exists a
  set $\R$ of linear length constraints implied by $w$ that is
  finite and complete.  Moreover, there is an algorithm which, given
  $w$, computes this set $\R$ of constraints.
\end{lemma}

\begin{proof} Since a word equation $w$ is
  logically equivalent to its solved form $\S$, every solution
  to $w$ is a solution to $\S$ and vice-versa. Hence, the set of
  length constraints implied by $w$ is equivalent to the set of length
  constraints implied by $\S$. 
  In $\R$, we will have integer variables associated with each string variable
  in $w$, integer variables associated with each unfixed part appearing in the 
  RHS of an equation in $\S$, and integer variables associated with each integer parameter
  appearing in the RHS of an equation in $\S$.
  For each $X$ appearing in $w$, consider the equation in $\S$ whose LHS is $X$:
  $X = t_1 \cdots t_n$, where each $t_i$ is either (1) a constant from $w$, 
  (2) a constant from $w$ raised to some integer parameter, or (3) an unfixed part.
  This equation implies a length equation of the form:
  $len(X) = C + i_1 c_1 + \cdots + i_k c_k + len(y_1) + \cdots len(y_j)$, 
  where $C$ is the sum of the lengths of constants in $w$ that appear on the RHS
  without an integer parameter; the $c_i$ terms are the lengths of constants
  with integer parameters; and there are terms for each unfixed part appearing in the 
  equation.  The only other length constraints associated with this equation 
  say that the unfixed parts and the integer parameters may be arbitrarily chosen: $i_r \geq 0$, $len(y_s) \geq 0$
  for each $1 \leq r \leq k$ and $1 \leq 1 \leq s \leq j$.
  Note that the minimum length of $X$ is the expression above 
  where we choose each $i_r = 0$ and each $len(y_s) = 0$.
  Let $\R$ be the union over $X$ in $w$ of the (finitely many) length constraints
  associated with $X$ discussed above.
  Since $\S$ is finite, so is $\R$.

  It remains to prove that $\R$ is complete. By definition of solved form, 
  all length constraints implied by $\S$ are of the form included in $\R$.
  Thus, $\R$ is complete for $\S$.
  Since $\S$ is logically equivalent with $w$, they 
  imply the same length constraints.  Hence, $\R$ is complete for $w$ as well.
\end{proof}

\begin{lemma}\label{lm:equiSAT}
  If a word equation $w$ has a solved form $\S$, then $w$ is
  equi-satisfiable with the length constraints $\R$ derived from $\S$.
\end{lemma}

\begin{proof} Since $\R$ is finite, the conjunction of all its elements 
is a formula of $\Lel{0}$ 

  ($\Rightarrow$) If $w$ is satisfiable, then so is $\R$: 
  Suppose $w$ is satisfiable and consider some satisfying assignment $w$.
  Then since $\R$ is implied by $w$, the lengths of the strings in this
  assignment satisfy all the constraints in $\R$.  Thus, this set of lengths
  witnesses the satisfiability of $\R$.

  ($\Leftarrow$) If $\R$ is satisfiable, then so is $w$: 
  Suppose $\R$ is satisfiable. Any solution of $\R$ gives 
  a collection of lengths for the variables in $w$.  An assignment 
  that satisfies $w$ is given by choosing arbitrary strings of the prescribed
  length for the unfixed parts and choosing values of the integer
  parameters prescribed by the solution of $\R$.
\end{proof}

\begin{thm}\label{th:Lel0}
  The satisifiability problem for $\Lel{0}$ formulas is decidable,
  provided that there is an algorithm to obtain the solved forms of
  word equations for which they exist.
\end{thm}

\begin{proof} 
  We assume without loss of generality that 
  the given $\Lel{0}$ formula is the conjunction of a single 
  word equation with some number of length constraints.  
  (Generalizing to arbitrary $\Lel{0}$ formulas is straightforward.)
  Let the input to the algorithm be a formula $\phi
  \wedge \theta$, where $\phi$ is the word equation and $\theta$ is a
  conjunction of length constraints. The output of the algorithm is
  {\it satisfiable} (SAT) or {\it unsatisfiable} (UNSAT).

  Plandowski's algorithm \cite{plandowski2006} decides 
  satisfiability of word equations; 
  known algorithms for formulas of Presburger arithmetic
  can decide the satisfiability of systems of linear length constraints.
  Thus, begin by running these algorithms (in parallel) 
  to decide if (separately) $\phi$ and $\theta$ are satisfiable.  
  If either of these return UNSAT, we
  return UNSAT.

  Using the assumption that the word equation $\phi$ has an
  effective solved form, compute this form $\S$ and the
  associated (complete and finite) implied set $\R$ of linear length
  constraints (as in Lemma \ref{lm:genR}).  By Lemma \ref{lm:equiSAT}, 
  it is now sufficient to check the satisfiability of $(\wedge \R) \wedge \theta$.
  This can be done by a second application of an algorithm 
  for formulas in Presburger arithmetic, because the length constraints
  implied by $\phi$ are all linear.  If this system of linear inequalities
  is satisfiable, return SAT, otherwise, we return UNSAT.

  This procedure is a sound, complete and terminating procedure for
  $\Lel{0}$-formulas whose word equations have effective solved
  forms.
\end{proof}

\subsection{\bf Practical Value of Solved Form and the Decidability
  Result}

JavaScript programs often process strings. These strings are entered into
 input forms on web-pages or are substrings used by JavaScript
programs to dynamically generate web-pages or SQL queries. 
During the processing of these
strings, JavaScript programs often concatenate these strings to form
larger strings, use strings in assignments, compare string lengths, construct equalities between strings
as part of if-conditionals or use regular expressions
as basic ``sanity-checks" of the strings being processed. Hence, any program
analysis of such JavaScript programs results in formulas that contain
string constants and variables, the concatenation operation, regular
expressions, word equations, and uses of the length function.

In their paper on an automatic JavaScript testing
program (Kudzu) and a practical satisfiability procedure for strings~\cite{prateek},
Saxena et al. mention generating more than 50,000 $\Lelr{0}$ formulas
where the length of the string variables is bounded (i.e., the string
variables range over a finite universe of strings). Kudzu takes as
input a JavaScript program and (implicit) specification, and does some
automatic analysis (a form of concrete and symbolic
execution~\cite{EXE,DART}) on the input program. The result of the
analysis is a string formula that captures the behavior of the
program-under-test in terms of the symbolic input to this program. A
solution of such a formula is a test input to the
program-under-test. Kudzu uses the Kaluza string solver to solve these
formulas and generate program inputs for program testing.

%

We obtained more than 50,000 string constraints (word equations +
length constraints) from the Kaluza team
(http://webblaze.cs.berkeley.edu/2010/kaluza/). Kaluza is a solver for
string constraints, where these constraints are obtained from
bug-finding and string analysis of web applications. The constraints
are divided into satisfiable and unsatisfiable constraints. We wrote a
simple Perl script to count the number of equations per file and the
number of equations already in solved form (identifier =
expression). We then computed the ratio to see how many examples from
this actual data set are already in solved form.

\subsubsection*{Experimental Results}

The results are divided into groups based on whether the 
constraints were satisfiable or not.  For satisfiable small
equations
(approximately 30-50 constraints per file), about $80 \%$ were 
already in solved form.  For satisfiable large equations (around 200 
constraints per file), 
this number rose to approximately $87\%$. 
Among the unsatisfiable and small equations (less than 20 constraints per file), 
again about $80\%$ were already in solved form. 
Large (greater than $4000$ constraints) unsatisfiable equations
were in solved form a slightly smaller percentage of the time: $75\%$.

\section{\bf Word Equations, Length, and Regular Expressions}
\label{sec:elr-decidability}

We now consider whether the previous result can be extended to show
that the satisfiability problem for $\Lelr{0}$ formulas is decidable,
provided that there is an algorithm to obtain the solved forms of
given word equations. A generalization of the proof strategy from
above looks promising.  That is, given a membership test in a regular
set $X \in RE$, we can extract from the structure of the regular
expression a constraint on the length of $X$ that is expressible as a
linear inequality.
 Thus, it  may seem that the same machinery as in the
$\Lel{0}$ theorem may be applied to the broader context of $\Lelr{0}$.
However, there remain some subtleties to resolve.  

\begin{examp}  Consider the $\Lelr{0}$ formula
\[
abX = Xba ~\wedge~X \in (ab)^*b ~\wedge~\text{len}(X) \leq 3.
\]
A na\"ive translation of each component into length constraints gives
us the following:
\[
\begin{cases}
  \text{len}(X) = 2i+1, i \geq 0 \qquad\text{\it implied by the word equation and regular expression} \\
  \text{len}(X) \leq 3.
\end{cases}
\]
This system of length constraints is easily seen to be simultaneously
satisfiable: let $i = 0$ or $1$ and hence $\text{len}(X) = 1$ or
$3$. However, the formula is {\bf not} satisfiable since solutions of the word
equation are $X \in (ab)^*a$ and the regular expression
requires any solution to end in a $b$.
\end{examp}

 Thus, in order to address
$\Lelr{0}$ formulas, we must take into account more information than
is encapsulated by the length constraints imposed by regular
expressions. In particular, if we impose the additional restriction
that the word equations must have solved form (without unfixed parts)
that are also regular expressions, then we can get a decidability
result for $\Lelr{0}$ formulas.

\begin{lemma}
\label{lem:re-lem}
If a word equation has a solved form without unfixed parts that is also a 
regular expression, then there is a finite set of linear length constraints
that can be effectively computed from this solved form and
which are equisatisfiable with the equation.
\end{lemma}

\begin{proof}
It is sufficient to recall the fact, from~\cite{Bl99}, that given a regular set $R$, 
the set of lengths of strings in $R$ is a
  finite union of arithmetic progressions.  Moreover, there is an
  algorithm to extract the parameters of these arithmetic progressions
  from the regular expression defining $R$.
\end{proof}

Using the above Lemma, the set of length constraints implied by an
arbitrary regular expression can be expressed as a finite system of
linear inequalities. 

\begin{thm}
  The satisifiability problem for $\Lelr{0}$ formulas is decidable,
  provided that there is an algorithm to obtain the solved forms of
  the given word equations, and the solved form equations do not
  contain unfixed parts and are regular expressions.
\end{thm}

The proof is a straightforward extension of the conditional
decidability proof given in Section~\ref{sec:decidability}.




  

%% file: relatedworkNeg.tex
\section{Related Work}
\label{sec:relwork}

In his original 1946 paper, Quine \cite{Quine} showed that the
first-order theory of string equations (i.e., quantified sentences
over Boolean combination of word equations) is undecidable.  Due to
the expressibility of many key reliability and verification questions
within this theory, this work has been extended in many ways.

One line of research studies fragments and modifications of this base
theory which are decidable.  Notably, in 1977, Makanin proved that the
satisfiability problem for the quantifier-free theory of word
equations is decidable \cite{makanin}.  In a sequence of papers,
Plandowski and co-authors showed that the complexity of this problem
is in PSPACE \cite{plandowski2006}.  Stronger results have been found
where equations are restricted to those where each variable occurs at
most twice\cite{robsondiekert} or in which there are at most two
variables \cite{CharaPach, IliePland, dabrowski2002weo}. In the first
case, satisfiability is shown to be NP-hard; in the second, polynomial
(which was improved further in the case of single variable word
equations).

Concurrently, many researchers have looked for the exact boundary
between decidability and undecidability. Durnev \cite{durnev} and
Marchenkov \cite{marchenkov} both showed that the $\forall\exists$
sentences over word equations is undecidable. Note that Durnev's
result is closest to our undecidability result. The main difference is
that our proof is considerably simpler because of the use of
two-counter machines, as opposed to certain non-standard machines used
by Durnev. We also note corollaries regarding number of occurences of
a variable, and $\Lel{1}$ sentences with a single universal followed
by bounded existentials. On the other hand, Durnev uses only 4 string
variables to prove his result, while we use 7. We believe that we can
reduce the number of variables, at the expense of a more complicated
proof.

Word equations augmented with additional predicates yield richer
structures which are relevant to many applications.  In the 1970s,
Matiyasevich formulated a connection between string equations
augmented with integer coefficients whose integers are taken from the
Fibonacci sequence and Diophantine equations \cite{Matiyasevich}.  In
particular, he showed that proving undecidability for the
satisfiability problem of this theory would suffice to solve Hilbert's
10th Problem in a novel way. Schulz \cite{schulz} extended Makanin's
satisfiability algorithm to the class of formulas where each variable
in the equations is specified to lie in a given regular set.  This is
a strict generalization of the solution sets of word equations.
\cite{KarhumakiPM97} shows that the class of sets expressible through
word equations is incomparable to that of regular sets.

M\"oller~\cite{Moller} studies word equations and related theories as
motivated by questions from hardware verification. More specifically,
M\"oller proves the undecidability of the existential fragment of a
theory of fixed-length bit-vectors, with a special finite but possibly
arbitrary concatenation operation, the extraction of substrings and
the equality predicate. Although this theory is related to the word
equations that we study, it is more powerful because of the
finite but possibly arbitrary concatenation.

%% file: strings.bbl
\begin{thebibliography}{10}

\bibitem{Bl99}
Achim Blumensath.
\newblock Automatic structures.
\newblock {Diploma thesis, RWTH-Aachen}, 1999.

\bibitem{EXE}
C.~Cadar, V.~Ganesh, P.M. Pawlowski, D.L. Dill, and D.R. Engler.
\newblock {EXE}: automatically generating inputs of death.
\newblock In Ari Juels, Rebecca~N. Wright, and Sabrina De~Capitani
  di~Vimercati, editors, {\em {ACM} Conference on Computer and Communications
  Security}, pages 322--335. ACM, 2006.

\bibitem{CharaPach}
W.~Charatonik and L.~Pacholski.
\newblock Word equations with two variables.
\newblock In H.~Abdulrab and J.-P. P{\'e}cuchet, editors, {\em IWWERT}, volume
  677 of {\em Lecture Notes in Computer Science}, pages 43--56. Springer, 1991.

\bibitem{dabrowski2002weo}
R.~Dabrowski and W.~Plandowski.
\newblock On word equations in one variable.
\newblock {\em Algorithmica}, 60(4):819--828, 2011.

\bibitem{durnev}
V.~Durnev.
\newblock Undecidability of the positive $\forall\exists^{3}$-theory of a free
  semigroup.
\newblock {\em Siberian Mathematical Journal}, 36(5):1067--1080, 1995.

\bibitem{EbFlumThomas}
H.-D. Ebbinghaus, J.~Flum, and W.~Thomas.
\newblock {\em Mathematical Logic}.
\newblock Undergraduate Texts in Mathematics. Springer-Verlag, 1994.

\bibitem{emmiMS2007}
M.~Emmi, R.~Majumdar, and K.~Sen.
\newblock Dynamic test input generation for database applications.
\newblock In D.S. Rosenblum and S.G. Elbaum, editors, {\em ISSTA}, pages
  151--162. ACM, 2007.

\bibitem{hampi}
V.~Ganesh, A.~Kiezun, S.~Artzi, P.J. Guo, P.~Hooimeijer, and M.D. Ernst.
\newblock {HAMPI}: A string solver for testing, analysis and vulnerability
  detection.
\newblock In G.~Gopalakrishnan and S.~Qadeer, editors, {\em CAV}, volume 6806
  of {\em Lecture Notes in Computer Science}, pages 1--19. Springer, 2011.

\bibitem{DART}
P.~Godefroid, N.~Klarlund, and K.~Sen.
\newblock {DART}: directed automated random testing.
\newblock In V.~Sarkar and M.W. Hall, editors, {\em PLDI}, pages 213--223. ACM,
  2005.

\bibitem{ullmanbook}
J.E. Hopcroft, R.~Motwani, and J.D. Ullman.
\newblock {\em Introduction to automata theory, languages, and computation}.
\newblock Pearson/Addison Wesley, 2007.

\bibitem{IliePland}
Lucian Ilie and Wojciech Plandowski.
\newblock Two-variable word equations.
\newblock {\em ITA}, 34(6):467--501, 2000.

\bibitem{KarhumakiPM97}
J.~Karhum{\"a}ki, F.~Mignosi, and W.~Plandowski.
\newblock The expressibility of languages and relations by word equations.
\newblock {\em J. ACM}, 47(3):483--505, 2000.

\bibitem{hampi2}
A.~Kiezun, V.~Ganesh, P.J. Guo, P.~Hooimeijer, and M.D. Ernst.
\newblock {HAMPI}: a solver for string constraints.
\newblock In G.~Rothermel and L.K. Dillon, editors, {\em ISSTA}, pages
  105--116. ACM, 2009.

\bibitem{rupak}
Rupak Majumdar.
\newblock Private correspondence.
\newblock SWS, MPI, Kaiserslautern, Germany, 2010.

\bibitem{makanin}
G.S. Makanin.
\newblock The problem of solvability of equations in a free semigroup.
\newblock {\em Math. Sbornik}, 103:147--236, 1977.
\newblock English transl. in Math USSR Sbornik 32 (1977).

\bibitem{marchenkov}
S.~S. Marchenkov.
\newblock Unsolvability of positive $\forall\exists$-theory of free semi-group.
\newblock {\em Sibirsky mathmatichesky jurnal}, 23(1):196--198, 1982.

\bibitem{Matiyasevich}
Yu. Matiyasevich.
\newblock Word equations, {Fibonacci} numbers, and {Hilbert}'s tenth problem.
\newblock Unpublished. Available at
  http://logic.pdmi.ras.ru/?yumat/Journal/jcontord.htm, 2006.

\bibitem{matiyasevich2008}
Yu. Matiyasevich.
\newblock Computation paradigms in light of {Hilbert's Tenth Problem}.
\newblock In S.B. Cooper, B.~L\"{o}we, and A.~Sorbi, editors, {\em New
  Computational Paradigms}, pages 59--85. Springer New York, 2008.

\bibitem{Moller}
Oliver M{\"o}ller.
\newblock $\exists {BV}_{[n]} solvability$.
\newblock Unpublished Manuscript. SRI International, Menlo Park, CA, USA,
  October 1996.

\bibitem{plandowski99}
W.~Plandowski.
\newblock Satisfiability of word equations with constants is in {PSPACE}.
\newblock In {\em FOCS}, pages 495--500. IEEE Computer Society, 1999.

\bibitem{plandowski2006}
W.~Plandowski.
\newblock An efficient algorithm for solving word equations.
\newblock In J.M. Kleinberg, editor, {\em STOC}, pages 467--476. ACM, 2006.

\bibitem{PRES27}
M.~Presburger.
\newblock {\"U}ber de vollst{\"a}ndigkeit eines gewissen systems der arithmetik
  ganzer zahlen, in welchen, die addition als einzige operation hervortritt.
\newblock In {\em Comptes Rendus du Premier Congr{\`e}s des Math{\'e}maticienes
  des Pays Slaves}, pages 92--101, 395, Warsaw, 1927.

\bibitem{Quine}
W.~V. Quine.
\newblock Concatenation as a basis for arithmetic.
\newblock {\em The Journal of Symbolic Logic}, 11(4):105--114, 1946.

\bibitem{robsondiekert}
J.M. Robson and V.~Diekert.
\newblock On quadratic word equations.
\newblock In C.~Meinel and S.~Tison, editors, {\em STACS}, volume 1563 of {\em
  Lecture Notes in Computer Science}, pages 217--226. Springer, 1999.

\bibitem{prateek}
P.~Saxena, D.~Akhawe, S.~Hanna, F.~Mao, S.~McCamant, and D.~Song.
\newblock A symbolic execution framework for {JavaScript}.
\newblock In {\em {IEEE} Symposium on Security and Privacy}, pages 513--528.
  {IEEE} Computer Society, 2010.

\bibitem{schulz}
K.~Schulz.
\newblock Makanin's algorithm for word equations-two improvements and a
  generalization.
\newblock In K.~Schulz, editor, {\em Word Equations and Related Topics}, volume
  572 of {\em Lecture Notes in Computer Science}, pages 85--150. Springer
  Berlin / Heidelberg, 1992.

\bibitem{WassermannSu2007}
G.~Wassermann and Z.~Su.
\newblock Sound and precise analysis of web applications for injection
  vulnerabilities.
\newblock In J.~Ferrante and K.S. McKinley, editors, {\em PLDI}, pages 32--41.
  ACM, 2007.

\end{thebibliography}
